\documentclass[twocolumn, amsmath,amssymb,aps, physrev]{revtex4-2}

\usepackage{graphicx}
\usepackage{dcolumn}
\usepackage{bm}

\usepackage{MnSymbol} 
\usepackage{color,xcolor}
\definecolor{color1}{HTML}{4e9773}
\definecolor{color1L}{HTML}{eef6f2}
\definecolor{color2}{HTML}{fab73d}
\definecolor{color2L}{HTML}{fffbf4}
\definecolor{color3}{HTML}{d2561a}
\definecolor{color3L}{HTML}{faede6}
\definecolor{color0}{HTML}{949494}
\definecolor{color0M}{HTML}{DCDCDC}
\definecolor{color0L}{HTML}{F8F8F8}

\usepackage{pgfplots}
\usepackage{graphicx}
\usepackage{tikz}
\usepackage{amsmath}
\usepackage{graphicx}
\usepgfplotslibrary{external} 
\usepackage{subcaption}
\usepackage{adjustbox} 		
\usepackage[section]{placeins} 
\usetikzlibrary{matrix,fit,calc,shapes.geometric,backgrounds,positioning,decorations.markings,decorations.pathreplacing,arrows,knots,hobby,angles,quotes,shapes,snakes,automata,petri}
\usepackage{wrapfig} 
\usepackage{hhline,colortbl} 
\usepackage[color=color2]{todonotes}

\usepackage{amsthm} 
\usepackage{amsmath,amssymb}
\usepackage{braket}
\usepackage{bbm}
\usepackage{nameref,cleveref}
\newtheorem{defi}{Definition}[section]

\newtheorem{prop}[defi]{Proposition}
\newtheorem{cor}[defi]{Corollary}
\crefname{chapter}{Chapter}{Chapter}
\Crefname{chapter}{Chapter}{Chapter}
\crefname{section}{Section}{Sections}
\Crefname{section}{Section}{Sections}
\crefname{algorithm}{Algorithm}{Algorithm}
\Crefname{algorithm}{Algorithm}{Algorithm}
\crefname{equation}{Equation}{Equation}
\Crefname{equation}{Equation}{Equation}
\crefname{figure}{Figure}{Figure}
\Crefname{figure}{Figure}{Figure}
\crefname{appendix}{Appendix}{Appendices}
\Crefname{appendix}{Appendix}{Appendices}
\crefname{prop}{Proposition}{Propositions}
\Crefname{prop}{Proposition}{Propositions}
\crefname{defi}{Definition}{Definitions}
\Crefname{defi}{Definition}{Definitions}
\newcommand{\tr}{\operatorname{tr}}

\begin{document}

\preprint{APS/123-QED}

\title{\textbf{Global zero-excitation state preparation through subsystem cooling} 
}%

\author{Kerstin Beer}
\email{Contact author: kerstin.beer@posteo.de}
\affiliation{
School of Mathematical and Physical Sciences, Macquarie University, Sydney, Australia
}%

\author{Daniel Burgarth}
\affiliation{%
Physics Department, Friedrich-Alexander Universit{\"a}t of Erlangen-Nuremberg, Staudtstr. 7, 91058 Erlangen, Germany
}%

\date{\today}

\begin{abstract}
Preparing interacting quantum systems in low-energy or ground states is a fundamental task in quantum simulation and quantum information processing. In realistic settings, dissipation and active cooling can typically be engineered only on a limited subset of the system. We study the dissipative dynamics of excitation-number conserving quantum systems governed by a GKLS master equation with local jump operators acting on a subset of qubits. We establish sufficient conditions under which such localized dissipation drives the full system to a unique globally attractive zero-excitation state. In particular, we prove that if the Hamiltonian generates excitation transfer described by a graph for which the dissipative subsystem forms a zero forcing set, then the zero-excitation state is the unique globally attractive stationary state. When this state coincides with a ground state of the Hamiltonian, the same mechanism realizes ground-state cooling. Our results provide a graph-theoretic criterion for global state preparation from localized dissipation, which we illustrate using a nearest-neighbor Heisenberg spin chain. For this model, a reduction to the single-excitation sector further yields a scaling estimate with the length of the chain which indicates efficient cooling.
\end{abstract}

\maketitle

\section{Introduction}

Quantum computing \cite{Nielsen2000} harnesses principles of quantum mechanics, such as superposition and entanglement, to process information in ways that differ fundamentally from classical computing. These properties enable quantum algorithms with advantages over their classical counterparts for specific tasks, including factorization \cite{Shor1994} and unstructured search \cite{Grover1996}. Every quantum algorithm requires a well-defined initial state from which the computation begins. Preparing well-controlled quantum states is therefore an essential ingredient in quantum computation, simulation, and other quantum technologies.

Many quantum information platforms rely on cooling in order to prepare low-energy or ground states. Even though quantum devices are often operated at extremely low temperatures, this does not necessarily prepare the system in the desired ground state, and residual excitations can remain. A variety of techniques have been developed for this purpose, ranging from laser cooling \cite{chan_laser_2011,stellmer_laser_2013,schreck_laser_2021} and evaporative cooling \cite{chaudhuri_evaporative_2007,li_tuning_2021} to adiabatic \cite{ho_intrinsic_2007,schachenmayer_adiabatic_2015}, measurement-based \cite{buffoni_quantum_2019,cotler_quantum_2019,langbehn_dilute_2024}, and dissipative approaches \cite{verstraete_quantum_2009,seki_dissipative_2026}. In large interacting systems, however, direct access to every constituent may be experimentally difficult or resource intensive.

\begin{figure*}
\centering
\begin{tikzpicture}[
    bn/.style={
        circle,
        fill=color3L,
        draw=black,
        font=\sffamily,
        minimum size=1mm
    },
    every node/.append style={bn},
    scale=0.7
]

\path node (1) {1}
    -- ++ (50:2.5) node (2) {2}
    -- ++(-95:1.75) node (3) {3}
    -- ++(-85:1.75) node (4) {4}
    -- ++(20:4.25) node (5) {5}
    -- ++(1.5,1.5) node (6) {6};

\begin{scope}[on background layer]
    \node[
        fill=color1L,
        draw=color1L,
        minimum size=4cm,
        text=black
    ] (c) at ([shift={(0.15,-1.5)}]6) {cooling};
\end{scope}

\draw (1)--(2)--(6)--(5)--(4)--(1)--(3)--(5)--(2)--(3)--(4);

\end{tikzpicture}
\caption{Schematic illustration of localized cooling applied to only part of a coupled qubit system. The Hamiltonian interactions can mediate the influence of the dissipative subsystem to qubits that are not cooled directly.}
\label{fig:cooling}
\end{figure*}
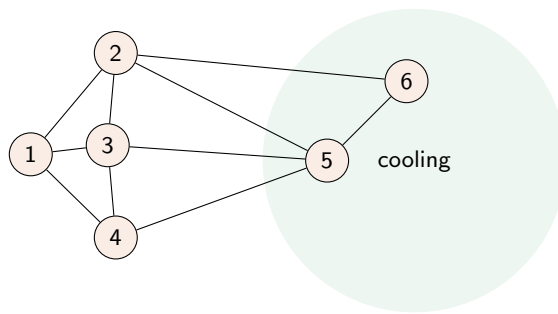

In this work, we investigate systems of coupled qubits in which dissipation acts only on a subset of the system, as illustrated schematically in Figure~\ref{fig:cooling}. This raises the central question of our work: under which conditions can localized dissipation drive an entire coupled quantum system to a unique zero-excitation state? When this zero-excitation state is also a ground state of the Hamiltonian, the same mechanism realizes ground-state cooling.

The possibility of controlling a large quantum system through access to only a small subsystem has previously been investigated in the context of quantum control. Burgarth and Giovannetti \cite{burgarth_full_2007} demonstrated a scheme in which local control of a subsystem is mediated to a larger system through a fixed coupling Hamiltonian. Related work developed an explicit protocol for cooling and controlling composite systems through local interactions \cite{burgarth_protocol_2008}. These protocols are formulated in terms of repeated discrete control steps, in which Hamiltonian evolution is interspersed with local operations such as swaps with auxiliary memory systems, whereas the present work considers continuous-time GKLS dynamics with Hamiltonian evolution and localized dissipation acting simultaneously. More recently, Langbehn \textit{et al.} \cite{langbehn_dilute_2024} studied measurement-induced cooling in the dilute limit and showed, for particular one-dimensional spin chains, that measurements applied to a single link can be sufficient for ground-state preparation. Related quasiparticle-cooling protocols have also used dissipative auxiliary degrees of freedom to remove excitations from interacting many-body systems and prepare low-energy states \cite{lloyd_quasiparticle_2025}. In contrast, rather than constructing a measurement or auxiliary-based cooling protocol for particular many-body models, we consider continuous-time GKLS dynamics with fixed localized jump operators and derive structural conditions under which such local dissipation guarantees convergence of the full system to a unique zero-excitation state. For the excitation-conserving qubit systems considered here, these conditions can be verified through the zero-forcing property of the excitation-transfer graph. In this sense, our criterion is not tied to one particular spin-chain model or cooling protocol: Hamiltonians with different coupling strengths and geometries can be treated within the same framework whenever they satisfy the required excitation-conservation and graph-connectivity conditions. This broader structural applicability comes at the cost of restricting the present analysis to excitation-conserving qubit systems and lowering-type dissipation. These results illustrate that access to only a small part of an interacting system can, under suitable conditions, influence the state of the full system.

Localized dissipation has also been studied directly in many-body open systems. For example, Kepesidis and Hartmann \cite{kepesidis_bose-hubbard_2012} considered a Bose--Hubbard model with particle loss acting on a single lattice site and showed that destructive interference can prevent part of the system from reaching the dissipative site. More generally, boundary-driven open quantum systems provide a broad setting in which dissipation acts only on restricted parts of an interacting system \cite{landi_nonequilibrium_2022,carlen_stationary_2025}. These examples also illustrate an important difficulty for localized cooling: excitations may remain trapped in portions of the system that are not sufficiently coupled to the dissipative degrees of freedom.

Complementary to these approaches, engineered dissipation has become an established tool for quantum-state preparation. Dissipative quantum-state engineering can be used to construct dynamics for which a desired target state is a stationary state or the unique stationary state \cite{kraus_preparation_2008,verstraete_quantum_2009,ding_single_2024,zhan_rapid_2026,ding_simple_2026}. Whereas many such approaches engineer or construct the dissipative dynamics specifically for the desired target state, our question is instead when a fixed localized lowering process acting only on a subsystem is already sufficient to determine the asymptotic state of the entire system. Structural properties of stationary states and their attractivity under GKLS dynamics have likewise been studied extensively. Baumgartner and Narnhofer \cite{baumgartner_analysis_2008} analyzed the structure of stationary states of finite-dimensional quantum dynamical semigroups with GKLS generators. As a special case, their results imply that for finite-dimensional GKLS dynamics, uniqueness of the steady state implies convergence to it. In the physics community, this was independently derived by Schirmer and Wang \cite{schirmer_stabilizing_2010}. Further sufficient conditions for uniqueness have been investigated in more recent work \cite{yoshida_uniqueness_2024}. These works address stationary-state structure and attractivity for broader classes of GKLS generators. Our setting is more specialized, but the additional excitation-conservation and locality structure allows us to reduce the question of uniqueness to a physically transparent condition on excitation transport and, for the qubit systems considered here, to the zero-forcing property of an interaction graph. These results provide the mathematical background for understanding when dissipative dynamics converges to a unique target state.

The speed of this convergence is another important aspect of dissipative state preparation. Spectral properties of quantum Markov generators can be used to characterize relaxation rates, and recent work has investigated spectral-gap bounds and mixing-time behavior in dissipative state preparation \cite{lucia_spectral_2025,zhan_rapid_2026}. While our primary question is the structural one of when dissipation acting only on part of a system is sufficient to enforce a unique global stationary state, for the Heisenberg chain considered below we additionally analyze the relaxation rate through a reduction to the single-excitation sector and derive its asymptotic system-size scaling.

Our approach separates this problem into two ingredients. First, we derive sufficient conditions for finite-dimensional GKLS dynamics under which a locally dissipative subsystem drives the full system to a unique globally attractive zero-excitation state. The conditions express conservation of total excitation number by the Hamiltonian, lowering of the excitation number by the dissipative jump operators, uniqueness of the locally annihilated state, and uniqueness of a Hamiltonian eigenstate compatible with the dissipative subsystem being in that state.

Second, for multipartite qubit systems with excitation-transfer interactions, we connect the last of these conditions to the structure of the interaction graph. Zero forcing is a graph-theoretic process in which an initially selected set of vertices progressively determines the rest of a graph through vertices with a unique unforced neighbor. It has previously been connected to controllability of linear and quantum systems on networks \cite{burgarth_zero_2013}. Here, we show that if the dissipative subsystem forms a zero forcing set of the excitation-transfer graph, then the global zero-excitation state is the only Hamiltonian eigenstate compatible with the dissipative subsystem being locally cooled. Zero forcing therefore provides a sufficient graph-theoretic condition for global attractivity under localized dissipation.

We illustrate these results using a nearest-neighbor Heisenberg spin chain, whose excitation-transfer graph is a path graph. Cooling one endpoint provides a zero forcing set, and the theoretical conditions therefore predict convergence to the global zero-excitation state. For small system sizes, we compare the full Liouvillian spectrum with a reduced single-excitation effective Hamiltonian and find agreement of the corresponding gap estimates to numerical precision. The reduced description further allows us to derive analytically an asymptotic $n^{-3}$ scaling of the reduced gap estimate, which we verify numerically for systems of up to $n=1000$ qubits.

The remainder of this paper is organized as follows. We first formulate the sufficient conditions for global convergence under subsystem-restricted GKLS dynamics and prove the corresponding proposition. We then introduce the excitation-transfer graph and zero forcing and derive a graph-theoretic corollary for multipartite qubit systems. Finally, we illustrate the theoretical results using the Heisenberg spin model, analyze the relaxation rate through the reduced single-excitation dynamics and its asymptotic scaling, and conclude with a discussion of the implications and possible directions for future work.

\section{Results}
To explore whether cooling a subset of qubits in a coupled quantum system can lead to the effective cooling of the entire system, we consider the dynamics generated by the Hamiltonian and describe the source of the cooling using GKLS dynamics. We divide the system into an indirectly cooled subsystem $A$ and an actively cooled subsystem $B$. We characterize cooling in terms of an excitation-counting operator $M$, such that the zero-excitation states of the two subsystems are denoted by $\ket{0_A}$ and $\ket{0_B}$. The following proposition gives conditions under which cooling subsystem $B$ drives the entire system to the zero-excitation state $\ket{0_A0_B}$.

\begin{prop}
\label{prop}
Let $\mathcal H = \mathcal H_A \otimes \mathcal H_B$ be a finite-dimensional Hilbert space.
Let
\begin{equation*}
M = M_A \otimes I_B + I_A \otimes M_B
\end{equation*}
be a Hermitian excitation-counting operator, where $\ket{0_A}$ and $\ket{0_B}$ are the unique eigenvectors of $M_A$ and $M_B$, respectively, associated with their smallest eigenvalues.

Suppose that the dynamics of a density operator $\rho$ is governed by the GKLS equation
\begin{equation}
\frac{d}{dt}\rho = \mathcal L(\rho)
= -i[H,\rho] + \sum_j \gamma_j \Big( B_j \rho B_j^\dagger - \tfrac{1}{2}\{B_j^\dagger B_j,\rho\} \Big),
\label{eqn:lindblad}
\end{equation}
where $\gamma_j > 0$, under the following conditions:
\begin{enumerate}

\item \textbf{Conservation of excitation number:} $M$ satisfies
\begin{equation*}
[H,M]=0.
\end{equation*}

\item \textbf{Lowering relation:} Each jump operator lowers the $M$-eigenvalue by $s_j > 0$:
\begin{equation*}
M B_j = B_j (M-s_j).
\end{equation*}

\item \textbf{Unique annihilated state:} The operators $B_j$ act only on subsystem $B$ and annihilate $\ket{0_B}$, i.e.
\begin{equation*}
B_j = I_A \otimes b_j,
\qquad b_j \ket{0_B}=0,
\end{equation*}
with
\begin{equation*}
\bigcap_j \ker(b_j)
=
\operatorname{span}\{\ket{0_B}\}.
\end{equation*}

\item \textbf{Factorized extremal eigenstate:}
There exists, up to a global phase, exactly one state
$\ket{\psi}$ of the form
\begin{equation*}
\ket{\psi}=\ket{\chi}\otimes\ket{0_B},
\qquad \ket{\chi}\in\mathcal H_A,
\end{equation*}
that is an eigenstate of $H$, i.e.
\begin{equation*}
H\ket{\psi}=h\ket{\psi}
\end{equation*}
for some $h\in\mathbb{R}$.

\end{enumerate}

Then the dynamics has a unique globally attractive stationary state:
\begin{equation*}
\lim_{t\to\infty} e^{t\mathcal L}(\rho(0))
=
\ket{0_A0_B}\bra{0_A0_B},
\qquad \forall \rho(0).
\end{equation*}
\end{prop}

\begin{proof}\hspace{1cm}

\textbf{Step 1 - Zero state is a stationary state:} 
As a first step, we check that the zero state
$\rho_0 := \ket{\psi_0}\bra{\psi_0}
=\ket{0_A0_B}\bra{0_A0_B}$
is a stationary state. Note that $\ket{0_A0_B}$ is a non-degenerate eigenvector of $M$ associated with its smallest eigenvalue. With Condition 1, $[H,M]=0$, we know that it is also an eigenvector of $H$. 
Furthermore, by Condition 3, $B_j\ket{0_A0_B}=0$ for all $j$. Therefore, both the Hamiltonian and dissipative parts of the GKLS equation vanish on $\rho_0$, and thus
\begin{equation*}
\mathcal L(\rho_0)=0.
\end{equation*}

\textbf{Step 2 - $\dfrac{d}{dt}\langle M\rangle$ under the lowering relation:}

In this step, we study the time derivative of the expectation value
$\dfrac{d}{dt}\langle M\rangle=\tr(M \frac{d \rho}{dt})$.
First, we consider an arbitrary state $\rho(t)$.

Using Equation~\ref{eqn:lindblad}, it follows
\begin{align*}
\frac{d}{dt}\langle M\rangle
=& \tr\!\left(M\mathcal L(\rho)\right)\\
=& -i\tr\!\left(M[H,\rho]\right)
+ \sum_j \gamma_j
\tr\!\big(
B_j^\dagger M B_j \rho\\
&-\frac{1}{2}M B_j^\dagger B_j\rho
-\frac{1}{2}B_j^\dagger B_jM\rho
\big).
\end{align*}

With Condition 1 and the cyclicity of the trace, it follows
\begin{equation*}
\tr(M[H,\rho])
=
\tr([M,H]\rho)
=
0.
\end{equation*}
Hence, with the lowering relation of Condition 2,
\begin{align*}
\frac{d}{dt}\langle M\rangle
=& \sum_j \gamma_j
\tr\!\left(
\left(
B_j^\dagger B_j(M-s_j)
-\frac{1}{2}M B_j^\dagger B_j
-\frac{1}{2}B_j^\dagger B_jM
\right)\rho
\right).
\end{align*}

Note that we use $[B_j^\dagger B_j,M]=0$: given
$M B_j=B_j(M-s_j)$, we obtain
$[M,B_j]=-s_jB_j$ and
$[M,B_j^\dagger]=s_jB_j^\dagger$.
Using $[M,AB]=[M,A]B+A[M,B]$, we get
\begin{align*}
[M,B_j^\dagger B_j]
&= [M,B_j^\dagger]B_j + B_j^\dagger[M,B_j]\\
&= s_jB_j^\dagger B_j - s_jB_j^\dagger B_j\\
&=0.
\end{align*}

Therefore,
\begin{align*}
\frac{d}{dt}\langle M\rangle
=& \sum_j \gamma_j
\tr\!\left(
\left(
B_j^\dagger B_jM
-s_jB_j^\dagger B_j
-\frac{1}{2}B_j^\dagger B_jM
-\frac{1}{2}B_j^\dagger B_jM
\right)\rho
\right)\\
=& -\sum_j s_j\gamma_j\,
\tr(B_j^\dagger B_j\rho)
\leq 0.
\end{align*}

Thus, the expectation value of $M$ cannot increase under the dynamics: the Hamiltonian evolution leaves it unchanged, while the dissipative terms can only decrease it.

Now, let us choose an arbitrary stationary state
$\mathcal L(\rho_\filledstar)=0$.
For a stationary state,
$\dfrac{d}{dt}\langle M\rangle=0$.
Since $s_j>0$, $\gamma_j>0$, and each summand
$\tr(B_j^\dagger B_j\rho_\filledstar)\geq0$
because $\rho_\filledstar\geq0$ and $B_j^\dagger B_j\geq0$,
it follows that
\begin{equation*}
\tr(B_j^\dagger B_j\rho_\filledstar)
=
0
\quad\text{for every }j.
\end{equation*}

Consider the spectral decomposition
$\rho_\filledstar=\sum_k p_k|\psi_k\rangle\langle\psi_k|$,
where $p_k>0$ and
$|\psi_k\rangle\in\mathrm{supp}(\rho_\filledstar)$.
We get the sum of non-negative terms
\begin{align*}
\tr(B_j^\dagger B_j\rho_\filledstar)
&=
\sum_k p_k
\langle\psi_k|B_j^\dagger B_j|\psi_k\rangle\\
&=
\sum_k p_k\,\|B_j|\psi_k\rangle\|^2.
\end{align*}
It follows that
$\|B_j|\psi_k\rangle\|^2=0$
for all $j$ and all
$|\psi_k\rangle\in\mathrm{supp}(\rho_\filledstar)$.
That means
$B_j|\psi_k\rangle=0$,
which we can phrase as
\begin{equation*}
\mathrm{supp}(\rho_\filledstar)
\subseteq
\ker(B_j).
\end{equation*}
Since we get this expression for all $j$, it follows overall that
\begin{equation*}
\mathrm{supp}(\rho_\filledstar)
\subseteq
\bigcap_j\ker(B_j).
\end{equation*}

Since $B_j=I_A\otimes b_j$, we have

\begin{equation*}
\bigcap_j\ker(B_j)
=
\mathcal H_A\otimes\bigcap_j\ker(b_j).
\end{equation*}

With Condition 3,
\begin{equation*}
\bigcap_j\ker(b_j)
=
\operatorname{span}\{\ket{0_B}\},
\end{equation*}
and therefore
\begin{equation*}
\mathrm{supp}(\rho_\filledstar)
\subseteq
\mathcal H_A\otimes
\operatorname{span}\{\ket{0_B}\}.
\end{equation*}

Hence, any stationary state must have the form
\begin{equation*}
\rho_\filledstar
=
\sigma\otimes\ket{0_B}\langle0_B|,
\end{equation*}
for some positive operator $\sigma$ on $\mathcal H_A$ with
$\tr(\sigma)=1$.

\textbf{Step 3 - Uniqueness from factorized extremal state:}

From Step 2, we know that any stationary state has the form
\begin{equation*}
\rho_\filledstar
=
\sigma\otimes\ket{0_B}\langle0_B|.
\end{equation*}
From Equation~\ref{eqn:lindblad}, we have
\begin{equation*}
-i[H,\rho_\filledstar]
=
\sum_j \gamma_j
\Big(
B_j\rho_\filledstar B_j^\dagger
-\tfrac12\{B_j^\dagger B_j,\rho_\filledstar\}
\Big).
\end{equation*}

Since $\rho_\filledstar=\sigma\otimes\ket{0_B}\langle0_B|$ and
$B_j=I_A\otimes b_j$ with $b_j\ket{0_B}=0$, the right-hand side vanishes, and thus
\begin{equation*}
[H,\rho_\filledstar]=0.
\end{equation*}

Since $H$ and $\rho_\filledstar$ are Hermitian and commute, they can be simultaneously diagonalized. Therefore, we can choose a common eigenbasis $\{\ket{\nu_i}\}$ such that
\begin{equation*}
\rho_\filledstar
=
\sum_i p_i \ket{\nu_i}\bra{\nu_i},
\end{equation*}
where every $\ket{\nu_i}$ with $p_i>0$ is also an eigenvector of $H$.

From Step 2, we already know that
\begin{equation*}
\mathrm{supp}(\rho_\filledstar)
\subseteq
\mathcal H_A\otimes\operatorname{span}\{\ket{0_B}\}.
\end{equation*}
Hence, every eigenvector $\ket{\nu_i}$ in the support of $\rho_\filledstar$ must have the form
\begin{equation*}
\ket{\nu_i}
=
\ket{\mu_i}\otimes\ket{0_B}.
\end{equation*}

With Condition 4, there exists exactly one eigenstate of $H$ of this form. Therefore, the support of $\rho_\filledstar$ is one-dimensional and
\begin{equation*}
\rho_\filledstar
=
\ket{\chi}\bra{\chi}
\otimes
\ket{0_B}\langle0_B|.
\end{equation*}

From Step 1, we already know that $\ket{0_A0_B}$ is an eigenstate of $H$. Since it is also a factorized state with subsystem $B$ in $\ket{0_B}$, Condition 4 implies that it is the unique eigenstate of this form, up to a global phase. Therefore,
\begin{equation*}
\rho_\filledstar
=
\ket{0_A0_B}\bra{0_A0_B}.
\end{equation*}

Thus, the stationary state is unique.

\textbf{Step 4 - Global attractivity:}

Since the GKLS dynamics acts on a finite-dimensional Hilbert space and $\rho_0$ is the unique stationary state, the result of Schirmer and Wang \cite{schirmer_stabilizing_2010} implies that $\rho_0$ is globally asymptotically stable. Therefore,
\begin{equation*}
\lim_{t\to\infty} e^{t\mathcal L}(\rho(0))
=
\rho_0
=
\ket{0_A0_B}\bra{0_A0_B}
\qquad
\forall \rho(0).
\end{equation*}

\end{proof}
In the following, we provide physical intuition for Proposition~\ref{prop}.
Condition 1 states that the Hamiltonian cannot create or destroy total excitations; it can only redistribute them between subsystems $A$ and $B$. Condition 2 guarantees that every dissipative event moves the system down the excitation ladder defined by $M$. Combined with Condition 1, this creates a directed flow in state space: the coherent dynamics redistributes excitations, while the dissipation removes them.

Condition 3 demands that $\ket{0_B}$ is the unique state of subsystem $B$ that is annihilated by all jump operators. Hence, once subsystem $B$ reaches $\ket{0_B}$, the dissipation acts trivially on it, and there is no other state of $B$ in which the dissipative dynamics can become trapped.

Condition 4 rules out dark states in which subsystem $B$ is already in its zero-excitation state while excitations remain trapped elsewhere in the system. Such a state would no longer be affected by the local dissipation and could therefore prevent convergence to the global zero-excitation state. Condition 4 excludes this possibility by requiring $\ket{0_A0_B}$ to be the only Hamiltonian eigenstate whose $B$-part is $\ket{0_B}$. We illustrate this mechanism more concretely below using simple excitation-transfer graphs in the discussion of zero forcing.

Proposition~\ref{prop} therefore shows that Conditions 1--4 are sufficient for convergence to the zero-excitation state. In a multipartite system, Condition 4 is the condition that is less direct to verify from the microscopic couplings. The remaining question is therefore whether the interaction structure of the Hamiltonian guarantees that an eigenstate with the dissipative subsystem in $\ket{0_B}$ must in fact have all qubits in their zero-excitation state. To formalize this idea, we introduce a graph-theoretic description of excitation transfer and relate it to the notion of zero forcing \cite{aim_zero_2008}.

\begin{defi}[zero forcing set]
Let $G=(V,E)$ be a graph. Given a $V_0 \subset V$ and starting from the initial coloring in which the vertices in $V_0$ are colored blue and all vertices in $V\setminus V_0$ are colored white, execute the following iterative step: if a blue vertex has exactly one white neighbor, then that neighbor is recolored blue.

The procedure is repeated until no further recoloring is possible. If this process results in all vertices being blue, then $V_0$ is called a \emph{zero forcing set} of $G$.
\end{defi}

Two simple examples illustrate the zero forcing rule.

\emph{Example 1: Chain.}
Consider the path graph
\begin{center}
\begin{tikzpicture}[every node/.style={circle,draw,minimum size=6mm,font=\small}]
\node[fill=blue!15] (1) at (0,0) {1};
\node (2) at (1.5,0) {2};
\node (3) at (3,0) {3};
\node (4) at (4.5,0) {4};
\draw (1)--(2)--(3)--(4);
\end{tikzpicture}
\end{center}
with $V_0=\{1\}$. Vertex $1$ forces $2$, then $2$ forces $3$, and finally $3$ forces $4$. Hence, $V_0=\{1\}$ is a zero forcing set.

\emph{Example 2: Y-shaped graph.}
Now consider the graph
\begin{center}
\begin{tikzpicture}[every node/.style={circle,draw,minimum size=6mm,font=\small}]
\node[fill=blue!15] (1) at (0,0) {1};
\node (2) at (1.5,0) {2};
\node (3) at (3,0.8) {3};
\node (4) at (3,-0.8) {4};
\draw (1)--(2);
\draw (2)--(3);
\draw (2)--(4);
\end{tikzpicture}
\end{center}

with $V_0=\{1\}$. Vertex $1$ first forces $2$. However, vertex $2$ then has two white neighbors, $3$ and $4$, so no further forcing step is possible. Hence, $V_0=\{1\}$ is not a zero forcing set.

These examples illustrate the role of branching: in the chain, there is a unique direction in which the forcing process can propagate, whereas in the Y-shaped graph the central vertex cannot determine which of several unforced neighbors should be forced next. In the quantum setting below, this distinction will correspond to whether localized dissipation can rule out excitations remaining elsewhere in the system.

To connect this graph-theoretic notion with the dynamics of our quantum system, we consider a system of $N$ qubits with vertex set $V=\{1,\ldots,N\}$ and Hilbert space
\begin{equation*}
\mathcal H = (\mathbb C^2)^{\otimes N}.
\end{equation*}
We denote by $\sigma_v^x, \sigma_v^y, \sigma_v^z$ the Pauli operators acting on qubit $v$, explicitly understood as acting trivially on all other qubits, i.e. with the identity acting on every qubit except $v$. We define the excitation-lowering and excitation-raising operators as
\begin{align*}
\sigma_v^- &= \frac{1}{2}\left(\sigma_v^x+i\sigma_v^y\right)
= \ket{0_v}\bra{1_v},\\
\sigma_v^+ &= \frac{1}{2}\left(\sigma_v^x-i\sigma_v^y\right)
= \ket{1_v}\bra{0_v}.
\end{align*}

The operators
$\sigma_u^+\sigma_v^-+\sigma_u^-\sigma_v^+$
exchange an excitation between qubits $u$ and $v$, whereas the terms
$\sigma_u^z\sigma_v^z$ and $\sigma_v^z$ do not change the excitation configuration.
Consequently, the connectivity through which excitations can propagate is determined by the nonzero transfer couplings $J_{uv}$.

\begin{defi}[Excitation-transfer graph for Hamiltonian $H$]
Let
\begin{equation*}
\mathcal H = (\mathbb C^2)^{\otimes N}
\end{equation*}
be the Hilbert space of $N$ qubits with vertex set $V=\{1,\ldots,N\}$.
Assume the Hamiltonian $H$ conserves excitation number, i.e. $[H,M]=0$, and can be written in the form
\begin{align*}
H =& \sum_{\{u,v\}\subset V}
J_{uv}\left(\sigma_u^+ \sigma_v^- + \sigma_u^- \sigma_v^+\right) \\
&+ \sum_{\{u,v\}\subset V}
\Delta_{uv}\, \sigma_u^z \sigma_v^z
+ \sum_{v\in V} h_v \sigma_v^z ,
\end{align*}
with real coefficients $J_{uv}, \Delta_{uv}, h_v$.

The \emph{excitation-transfer graph} associated with $H$ is the undirected graph
$G=(V,E)$ with vertex set $V$ and edge set
\begin{equation*}
E = \{\{u,v\}\subset V \mid J_{uv} \neq 0\}.
\end{equation*}
\end{defi}

Using this graph-theoretic perspective, we can heuristically interpret how the influence of the dissipative subsystem propagates through the excitation-transfer graph induced by the Hamiltonian. Suppose a subset of qubits forms the dissipative subsystem and is therefore driven towards its zero-excitation state. If a qubit in this set has exactly one neighboring qubit outside the set, then an excitation on that neighboring qubit would couple, through the Hamiltonian, to the already cooled qubit. For an eigenstate whose cooled qubits remain in their zero-excitation state, this is only possible if the neighboring qubit is also in its zero-excitation state. Consequently, this qubit can be added to the effectively cooled set. Repeating this argument iteratively extends the set of qubits constrained to the zero-excitation state. If the initially cooled subsystem forms a zero forcing set of the excitation-transfer graph, the argument eventually reaches the entire system. Hence, the only Hamiltonian eigenstate compatible with the dissipative subsystem being in its zero-excitation state is the global zero-excitation state. By Proposition~\ref{prop}, this implies convergence to that state. The above intuition is formalized in the following corollary.

\begin{cor}
\label{cor}

Let
\begin{equation*}
\mathcal H = (\mathbb C^2)^{\otimes N}
\end{equation*}
be a system of $N$ qubits, and assume that Conditions 1, 2 and 3 of Proposition~\ref{prop} are satisfied.
Let $G=(V,E)$ denote the excitation-transfer graph associated with $H$, and let $V_0 \subset V$ be the set of vertices corresponding to the dissipative subsystem $B$.

If $V_0$ is a zero forcing set of $G$, then Condition 4 of Proposition~\ref{prop} is satisfied.
Hence, the GKLS dynamics has a unique globally attractive stationary state given by the zero-excitation state
\begin{equation*}
\rho_0 = |0_A0_B\rangle\langle 0_A0_B|.
\end{equation*}
In particular,
\begin{equation*}
\lim_{t\to\infty} e^{t\mathcal L}(\rho) = \rho_0
\qquad \text{for all initial states } \rho .
\end{equation*}

\end{cor}

\begin{proof}

We show that the zero forcing condition implies Condition 4 of
Proposition~\ref{prop}. Let $\ket{\psi}$ be an eigenstate of $H$ of the form
\begin{equation*}
\ket{\psi}
=
\ket{\chi}\otimes\ket{0_B},
\qquad
H\ket{\psi}=h\ket{\psi}.
\end{equation*}
Since the vertices in $V_0$ correspond to subsystem $B$, all vertices in
$V_0$ are in their zero-excitation state. Hence, we can write
\begin{equation*}
\ket{\psi}
=
\ket{0_{V_0}}\otimes\ket{\phi}_{V\setminus V_0}.
\end{equation*}

Suppose that, at some stage of the zero forcing process, all vertices in a
set $S$ with
\begin{equation*}
V_0\subseteq S\subseteq V
\end{equation*}
are known to be in their zero-excitation state, such that
\begin{equation*}
\ket{\psi}=\ket{0_S}\otimes\ket{\phi}_{V\setminus S}.
\end{equation*}
Let $v\in S$ be a vertex with exactly one neighbor $w\notin S$, as required
by the zero forcing rule.
Consider the eigenvalue equation
\begin{equation*}
H\ket{\psi}=h\ket{\psi}
\end{equation*}
and project it onto the subspace in which $v$ is excited while all other
vertices in $S$ remain in their zero-excitation state. The projection of the
right-hand side vanishes, since $v$ is in the state $\ket{0_v}$ in
$\ket{\psi}$.
The diagonal terms
$\sigma_u^z\sigma_v^z$ and $\sigma_v^z$ in $H$ do not change the excitation
configuration and therefore do not contribute to this projected subspace.
Likewise, excitation-transfer terms not incident on $v$ cannot excite $v$,
and excitation-transfer terms between $v$ and vertices already in $S$
vanish because these vertices are in their zero-excitation state.\\
Since $w$ is the unique neighbor of $v$ outside $S$, the only remaining term
that can create an excitation on $v$ is
\begin{equation*}
J_{vw}\sigma_v^+\sigma_w^- .
\end{equation*}
Therefore, the projected eigenvalue equation gives
\begin{equation*}
J_{vw}\sigma_w^-\ket{\phi}_{V\setminus S}=0.
\end{equation*}
Since $\{v,w\}\in E$, we have $J_{vw}\neq0$, and hence
\begin{equation*}
\sigma_w^-\ket{\phi}_{V\setminus S}=0.
\end{equation*}
For a qubit,
\begin{equation*}
\ker(\sigma_w^-)
=
\operatorname{span}\{\ket{0_w}\},
\end{equation*}
and therefore $w$ must also be in its zero-excitation state. Thus,
\begin{equation*}
\ket{\psi}
=
\ket{0_{S\cup\{w\}}}
\otimes
\ket{\phi'}_{V\setminus(S\cup\{w\})}.
\end{equation*}
Repeating this argument along the zero forcing sequence, and using that
$V_0$ is a zero forcing set, eventually all vertices in $V$ are forced to
their zero-excitation state. Hence,
\begin{equation*}
\ket{\psi}
=
\ket{0_1\cdots0_N}
=
\ket{0_A0_B}.
\end{equation*}
Thus, up to a global phase, $\ket{0_A0_B}$ is the unique eigenstate of $H$
of the form $\ket{\chi}\otimes\ket{0_B}$. Therefore, Condition 4 of
Proposition~\ref{prop} is satisfied.

By Proposition~\ref{prop}, the GKLS dynamics therefore has the unique
globally attractive stationary state
\begin{equation*}
\rho_0
=
\ket{0_A0_B}\bra{0_A0_B},
\end{equation*}
and
\begin{equation*}
\lim_{t\to\infty}e^{t\mathcal L}(\rho)
=
\rho_0
\qquad
\text{for all initial states }\rho.
\end{equation*}

\end{proof}

The inductive argument above is closely related to the propagation mechanism used in the zero-forcing controllability framework of \cite{burgarth_zero_2013}. In both settings, the zero forcing process provides an ordering in which influence propagates from an initial set of vertices to the rest of the graph via vertices with a unique unforced neighbor. The difference is in the physical interpretation: in \cite{burgarth_zero_2013}, zero forcing is used to show the propagation of controllability through the network, whereas here it is used to show that no Hamiltonian eigenstate can retain excitations outside the dissipative subsystem while the latter is in its zero-excitation state. To see this explicitly, consider again the Y-shaped graph in Example 2 above. Denoting states with a single excitation at site $n$ by $|n\rangle$, it is easy to see that $|3\rangle -|4\rangle$ is an eigenstate of the Hamiltonian with eigenvalue $0$, a dark state, which is also invariant under the noise and therefore remains stuck.  In this way, zero forcing provides a sufficient condition for Condition 4 of Proposition~\ref{prop}, which then implies convergence to the global zero-excitation state under localized dissipation.

\section{Numerical simulations}

For numerical results supporting our findings above, we consider the Heisenberg spin model \cite{auerbach_interacting_1994}, which describes a system of spins interacting via nearest-neighbor exchange interactions. This model plays a central role in quantum many-body physics, capturing essential features of magnetism, entanglement, and quantum phase transitions.

In the following, we denote the spin-$1/2$ operators by
$S^\alpha=\sigma^\alpha/2$ with $\alpha\in\{x,y,z\}$, where $\sigma^\alpha$ are the Pauli matrices, and define
\[
S_i^\alpha=\mathbbm{1}^{\otimes i-1}\otimes S^\alpha\otimes\mathbbm{1}^{\otimes n-i}.
\]
The Hamiltonian used in the numerical simulations is
\begin{equation*}
H = \sum_{i=1}^{n-1} \sum_{\alpha=x,y,z} S_i^\alpha S_{i+1}^\alpha.
\end{equation*}

Hence, for every nearest-neighbor pair $\{i,i+1\}$, the Hamiltonian induces coherent excitation transfer between vertices $i$ and $i+1$.

The excitation-transfer graph associated with $H$ therefore coincides with the path graph $P_n$. Any endpoint of the path graph is a zero forcing set; in particular, choosing the last vertex $v_B=n$ gives the zero forcing sequence
\begin{equation*}
n\rightarrow n-1\rightarrow\cdots\rightarrow1.
\end{equation*}

We define the excitation-counting operator as
\begin{equation*}
M=\sum_{i=1}^{n}\left(\frac{\mathbbm{1}}{2}-S_i^z\right)=M_A+M_B,
\end{equation*}
where
\begin{equation*}
M_A=\sum_{i=1}^{n-1}\left(\frac{\mathbbm{1}}{2}-S_i^z\right),
\qquad
M_B=\frac{\mathbbm{1}}{2}-S_n^z.
\end{equation*}
With this convention, $\ket{0}^{\otimes n}$ is the unique zero-excitation state.

We choose
\begin{equation*}
B=\mathbbm{1}^{\otimes n-1}\otimes\ket{0}\bra{1}.
\end{equation*}
We denote the corresponding dissipation rate by $\gamma_1$ and set
$\gamma_1=1$ in the numerical simulations unless stated otherwise.

To see that $[H,M]=0$, note that
\begin{equation*}
S_i^xS_{i+1}^x+S_i^yS_{i+1}^y
=
\frac{1}{2}\left(
S_i^+S_{i+1}^-+S_i^-S_{i+1}^+
\right),
\end{equation*}
where $S_i^\pm=S_i^x\mp iS_i^y$.
These terms exchange an excitation between neighboring qubits without changing the total excitation number, while the terms $S_i^zS_{i+1}^z$ leave the excitation configuration unchanged. Hence,
\begin{equation*}
[H,M]=0.
\end{equation*}
Furthermore, the jump operator lowers the excitation number by one,
\begin{equation*}
MB=B(M-1),
\end{equation*}
and
\begin{equation*}
\ker(\ket{0}\bra{1})
=
\operatorname{span}\{\ket{0}\}.
\end{equation*}
Thus, Conditions 1--3 of Proposition~\ref{prop} are satisfied. Since the dissipative vertex $v_B=n$ is a zero forcing set of the excitation-transfer graph, Corollary~\ref{cor} applies. Therefore, the global zero-excitation state $\ket{0}^{\otimes n}$ is the unique globally attractive stationary state.

To simulate the time evolution of a density matrix
\begin{equation*}
\frac{d\rho}{dt} = \mathcal{L}(\rho),
\end{equation*}
we vectorize $\rho$ and denote the matrix representation of the Liouvillian
$\mathcal{L}$ by $\widetilde{\mathcal{L}}$. The vectorized density matrix is
then propagated between consecutive time steps according to
\begin{equation}
|\rho(t+\delta t))
=
e^{\widetilde{\mathcal{L}}\delta t}
|\rho(t)).
\end{equation}

We denote the eigenvalues of the Liouvillian $\mathcal{L}$ by
$\{\lambda_k\}\subset\mathbb{C}_{\le 0}$. For the dissipative dynamics considered here,
\begin{equation*}
\operatorname{Re}(\lambda_k)\leq0,
\end{equation*}
and the stationary state corresponds to the eigenvalue $\lambda_0=0$.

For a unique stationary state, we define the positive Liouvillian spectral gap as
\begin{equation*}
\Delta
=
-\max_{\lambda_k\neq0}\operatorname{Re}(\lambda_k).
\end{equation*}
Thus, $\Delta$ is determined by the nonzero Liouvillian eigenvalue whose real part lies closest to zero.

\begin{figure}[h]
\centering
\begin{subfigure}[b]{0.32\textwidth}
\centering
\includegraphics[width=\linewidth]{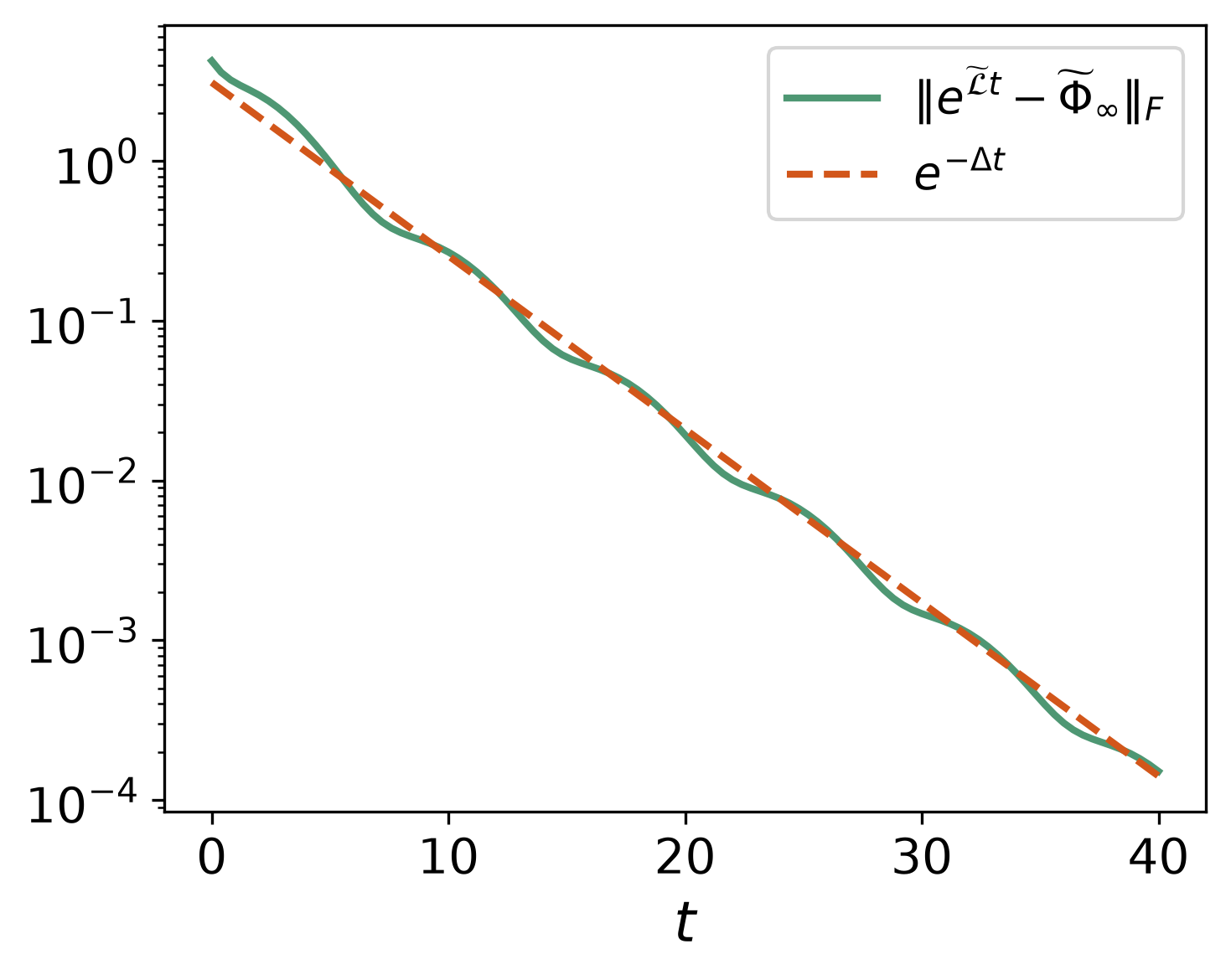}
\caption{$n=2$}
\end{subfigure}
\begin{subfigure}[b]{0.32\textwidth}
\centering
\includegraphics[width=\linewidth]{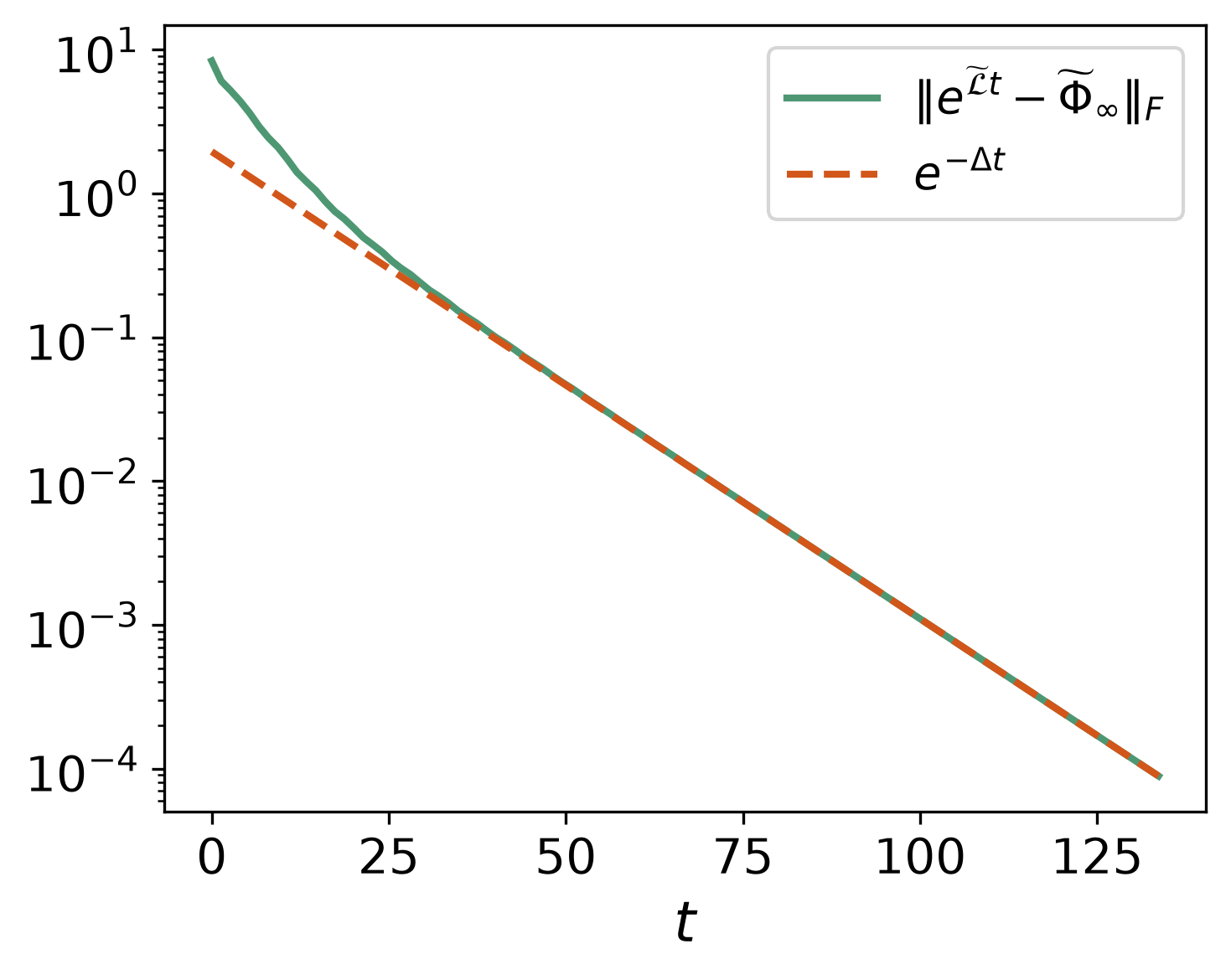}
\caption{$n=3$}
\end{subfigure}
\begin{subfigure}[b]{0.32\textwidth}
\centering
\includegraphics[width=\linewidth]{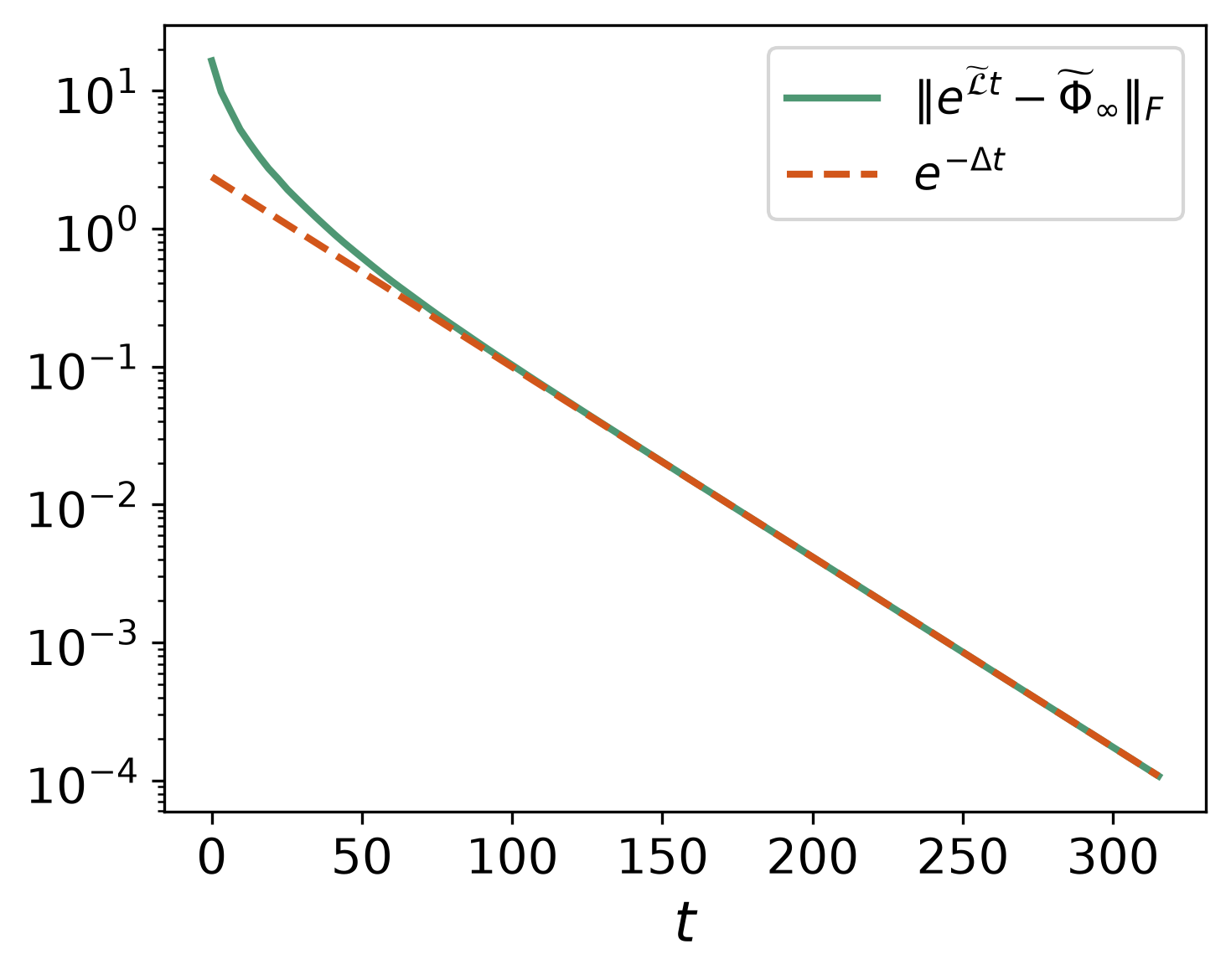}
\caption{$n=4$}
\end{subfigure}

\caption{\textbf{Convergence of the dynamical map to its asymptotic cooling map.}
The solid curves show the Frobenius norm
$\|e^{\widetilde{\mathcal{L}}t}-\widetilde{\Phi}_\infty\|_F$
for systems of (a) two, (b) three, and (c) four qubits. The vertical axis is logarithmic. The dashed curves decay as $e^{-\Delta t}$, with $\Delta$ obtained from the Liouvillian spectrum.}
\label{fig:norm}
\end{figure}

This quantity determines the asymptotic convergence timescale of the system
toward the stationary state \cite{spohn_algebraic_1977}, with the slowest
decaying contribution scaling as
\begin{equation*}
e^{-\Delta t}.
\end{equation*}

For the cooling dynamics considered here, the unique stationary state is
\begin{equation*}
\rho_{\mathrm{ss}}
=
\ket{0\ldots0}\bra{0\ldots0}.
\end{equation*}

We compare the full dynamical map with its asymptotic limit. Since every density matrix has unit trace and converges to $\rho_{\mathrm{ss}}$, the asymptotic map acting on an arbitrary operator $X$ is
\begin{equation*}
\Phi_\infty(X)
=
\rho_{\mathrm{ss}}\operatorname{Tr}(X).
\end{equation*}

The factor $\operatorname{Tr}(X)$ extends the map linearly from normalized density matrices to the full operator space.

Using the same vectorization convention as above, the matrix representation of this map is
\begin{equation*}
\widetilde{\Phi}_\infty
=
|\rho_{\mathrm{ss}}))((\mathbbm{1}|,
\end{equation*}
where
\begin{equation*}
((\mathbbm{1}|X))
=
\operatorname{Tr}(X).
\end{equation*}

We therefore evaluate
\begin{equation*}
\left\|
e^{\widetilde{\mathcal{L}}t}
-
\widetilde{\Phi}_\infty
\right\|_F,
\end{equation*}
where for simplicity we chose the Frobenius norm,
\begin{equation*}
\|A\|_F
=
\sqrt{\operatorname{Tr}\left(A^\dagger A\right)}.
\end{equation*}
This quantity vanishes when the dynamical map approaches its asymptotic cooling map.

Figure~\ref{fig:norm} shows this convergence for $n=2,\ldots,4$ on a logarithmic vertical axis. The dashed reference curves decay as $e^{-\Delta t}$, where $\Delta$ is obtained from the Liouvillian spectrum. Their vertical normalization is fixed using a late-time point of the corresponding numerical curve. The numerical curves approach the predicted asymptotic decay rate at late times. The corresponding simulations and plotting code are available in the accompanying repository~\cite{beer_cooling_code}.

\section{Single-excitation reduction and asymptotic scaling}
Let us now try to get an analytical handel on the scaling of the cooling efficiency for large dimensions. Intuitively, the dominant contribution (the hardest to catch state) should stem from the single excitation sector, where the probabiliy to catch an excitation on the cooling site would be lowest.
For $n=2,\ldots,8$, this intuition is confirmed numerically, as the reduced single-excitation gap
$\Delta_{\mathrm{red}}$ introduced below agrees with a direct calculation
of the full Liouvillian gap $\Delta$ to numerical precision. This
motivates us to use the single-excitation sector to analyze the relaxation
rate for larger systems and to investigates its asymptotic scaling.

Since the Hamiltonian conserves the total excitation number, the Hilbert space decomposes into fixed-excitation sectors,
\begin{equation*}
\mathcal{H}
=
\bigoplus_{m=0}^{n}\mathcal{H}_m,
\end{equation*}
where $\mathcal{H}_m$ denotes the subspace containing exactly $m$ excitations. The zero-excitation sector is
\begin{equation*}
\mathcal{H}_0
=
\operatorname{span}\{\ket{\mathrm{vac}}\},
\qquad
\ket{\mathrm{vac}}
=
\ket{0}^{\otimes n},
\end{equation*}
and the single-excitation sector is
\begin{equation*}
\mathcal{H}_1
=
\operatorname{span}\{\ket{j}\}_{j=1}^{n},
\end{equation*}
where $\ket{j}$ denotes the state with a single excitation localized at site $j$.

Restricting $H$ to $\mathcal{H}_1$ gives the $n\times n$ Hamiltonian
\begin{equation*}
H_1
=
\begin{pmatrix}
\frac{n-3}{4} & \frac{1}{2} & 0 & \cdots & 0 \\
\frac{1}{2} & \frac{n-5}{4} & \frac{1}{2} & \ddots & \vdots \\
0 & \frac{1}{2} & \ddots & \ddots & 0 \\
\vdots & \ddots & \ddots & \frac{n-5}{4} & \frac{1}{2} \\
0 & \cdots & 0 & \frac{1}{2} & \frac{n-3}{4}
\end{pmatrix}.
\end{equation*}
The off-diagonal entries $1/2$ arise from the
$S_i^xS_{i+1}^x+S_i^yS_{i+1}^y$ terms and describe hopping of the excitation between neighboring sites. The diagonal terms arise from the $S_i^zS_{i+1}^z$ interactions. A boundary excitation creates one anti-aligned nearest-neighbor bond, giving the diagonal value $(n-3)/4$, whereas a bulk excitation creates two, giving $(n-5)/4$.

Within the zero- and single-excitation sectors, the cooling jump acts as
\begin{equation*}
B
=
\ket{\mathrm{vac}}\bra{n},
\end{equation*}
and hence
\begin{equation*}
B^\dagger B
=
\ket{n}\bra{n}.
\end{equation*}
The corresponding non-Hermitian effective Hamiltonian in the single-excitation sector is therefore
\begin{equation*}
H_{\mathrm{eff}}
=
H_1
-
\frac{i\gamma_1}{2}\ket{n}\bra{n}.
\end{equation*}
The imaginary term represents loss of amplitude when the excitation occupies the dissipative boundary site.

We denote the eigenvalues of $H_{\mathrm{eff}}$ by $\{\mu_a\}$. Let
$\ket{\phi_a}$ be a corresponding right eigenvector. Since the vacuum is
unaffected by the jump operator and has energy
\begin{equation*}
E_{\mathrm{vac}}
=
\frac{n-1}{4},
\end{equation*}
the operator $\ket{\phi_a}\bra{\mathrm{vac}}$ is an eigenoperator of the
Liouvillian with eigenvalue
\begin{equation*}
\lambda_a^{\mathrm{red}}
=
-i\left(\mu_a-E_{\mathrm{vac}}\right).
\end{equation*}
Its decay rate is therefore
\begin{equation*}
-\operatorname{Re}\left(\lambda_a^{\mathrm{red}}\right)
=
-\operatorname{Im}(\mu_a).
\end{equation*}
We define the reduced single-excitation gap as
\begin{equation*}
\Delta_{\mathrm{red}}(n)
=
\min_a\left[-\operatorname{Im}(\mu_a)\right].
\end{equation*}

The reduction is also computationally advantageous. Whereas the full
Hilbert-space dimension grows as $2^n$ and the matrix representation of the
Liouvillian has dimension $4^n$, $H_{\mathrm{eff}}$ is only an
$n\times n$ tridiagonal matrix. We therefore use
$\Delta_{\mathrm{red}}(n)$ to study substantially larger systems, as well as providing analytical asymptotics.

The asymptotic system-size dependence of $\Delta_{\mathrm{red}}(n)$ can be
obtained analytically. To expose the quasi-uniform tridiagonal structure,
we define
\begin{equation*}
T_n
=
2H_{\mathrm{eff}}
-
\frac{n-5}{2}\mathbbm{1}
=
\begin{pmatrix}
1 & 1 & 0 & \cdots & 0 \\
1 & 0 & 1 & \ddots & \vdots \\
0 & 1 & \ddots & \ddots & 0 \\
\vdots & \ddots & \ddots & 0 & 1 \\
0 & \cdots & 0 & 1 & 1-i\gamma_1
\end{pmatrix}.
\end{equation*}
Following the treatment of quasi-uniform tridiagonal matrices in
Ref.~\cite{banchi_spectral_2013}, we parameterize the bulk eigenvalues as
\begin{equation*}
\theta=2\cos k.
\end{equation*}
For an eigenvector $v=(v_1,\ldots,v_n)^T$ of $T_n$ with eigenvalue
$\theta=2\cos k$, the bulk eigenvalue equation is
\begin{equation*}
v_{j-1}+v_{j+1}
=
2\cos(k)\,v_j,
\qquad
j=2,\ldots,n-1.
\end{equation*}
At the left boundary, the first row of $T_n$ gives
\begin{equation*}
v_1+v_2
=
2\cos(k)\,v_1.
\end{equation*}
Comparing this with the bulk recurrence extended formally to $j=1$
shows that the left boundary condition can be written as
\begin{equation*}
v_0=v_1.
\end{equation*}
A solution of the bulk recurrence satisfying this condition is therefore
\begin{equation*}
v_j
\propto
\cos\left[\left(j-\frac{1}{2}\right)k\right].
\end{equation*}

At the dissipative boundary, the last row of $T_n$ gives
\begin{equation*}
v_{n-1}
+
(1-i\gamma_1)v_n
=
2\cos(k)\,v_n.
\end{equation*}

Equivalently, extending the bulk recurrence to a fictitious site $n+1$
gives the right boundary condition
\begin{equation*}
v_{n+1}
=
(1-i\gamma_1)v_n.
\end{equation*}
Substituting the cosine solution yields the spectral condition
\begin{equation*}
\cos\left[\left(n+\frac{1}{2}\right)k\right]
=
(1-i\gamma_1)
\cos\left[\left(n-\frac{1}{2}\right)k\right].
\end{equation*}
This can equivalently be written as
\begin{equation*}
\tan(nk)
=
\frac{i\gamma_1}{2-i\gamma_1}
\cot\left(\frac{k}{2}\right).
\end{equation*}

\begin{figure}[h]
\centering
\includegraphics[width=0.85\linewidth]{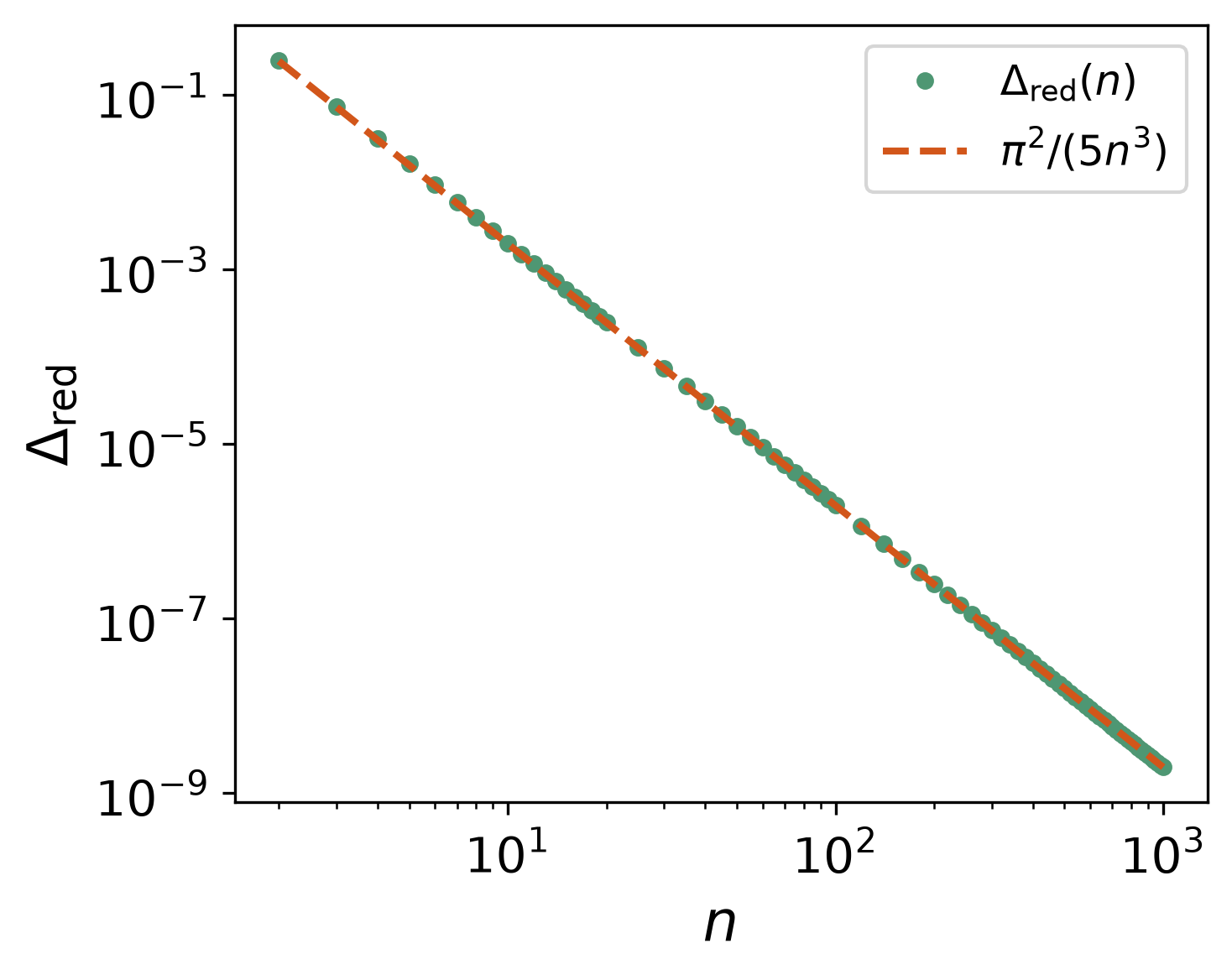}
\caption{\textbf{System-size scaling of the single-excitation gap estimate.}
The points show $\Delta_{\mathrm{red}}(n)$ obtained numerically from the
reduced single-excitation effective Hamiltonian for systems up to
$n=1000$ with $\gamma_1=1$. Both axes are logarithmic. The dashed line
shows the corresponding analytical asymptotic result
$\Delta_{\mathrm{red}}(n)=\pi^2/(5n^3)$.}
\label{fig:gap_scaling}
\end{figure}

For fixed $\gamma_1>0$, the least-damped modes approach the two band edges
as $n$ increases. Expanding the spectral condition around the lower band
edge, $k\rightarrow\pi$, gives
\begin{equation*}
-\operatorname{Im}\theta_{\mathrm{lower}}
=
\frac{2\gamma_1}{4+\gamma_1^2}
\frac{\pi^2}{n^3}
+
O(n^{-4}),
\end{equation*}
whereas expansion around the upper band edge, $k\rightarrow0$, gives
\begin{equation*}
-\operatorname{Im}\theta_{\mathrm{upper}}
=
\frac{1}{2\gamma_1}
\frac{\pi^2}{n^3}
+
O(n^{-4}).
\end{equation*}
The least-damped mode is determined by the smaller of these two decay
rates, so that
\begin{equation*}
-\operatorname{Im}\theta_{\mathrm{slow}}
=
\frac{\pi^2}{n^3}
\min\left\{
\frac{2\gamma_1}{4+\gamma_1^2},
\frac{1}{2\gamma_1}
\right\}
+
O(n^{-4}).
\end{equation*}

Since
\begin{equation*}
\theta
=
2\mu-\frac{n-5}{2},
\end{equation*}
the shift is purely real, and hence
\begin{equation*}
\Delta_{\mathrm{red}}(n)
=
-\operatorname{Im}\mu_{\mathrm{slow}}
=
-\frac{1}{2}\operatorname{Im}\theta_{\mathrm{slow}}.
\end{equation*}
Consequently,
\begin{equation}
\Delta_{\mathrm{red}}(n)
=
\frac{\pi^2}{n^3}
\min\left\{
\frac{\gamma_1}{4+\gamma_1^2},
\frac{1}{4\gamma_1}
\right\}
+
O(n^{-4}).
\label{eq:gap_asymptotic}
\end{equation}

For the value $\gamma_1=1$ used in our numerical simulations, the
lower-band contribution is the smaller one, and Eq.~\eqref{eq:gap_asymptotic}
reduces to
\begin{equation*}
\Delta_{\mathrm{red}}(n)
=
\frac{\pi^2}{5n^3}
+
O(n^{-4}).
\end{equation*}

Figure~\ref{fig:gap_scaling} compares the analytical asymptotic result for
$\gamma_1=1$ with the numerical eigenvalues of $H_{\mathrm{eff}}$ for
systems up to $n=1000$. The numerical results approach the predicted
$n^{-3}$ scaling and its analytical prefactor as the system size increases.

Thus, for any fixed $\gamma_1>0$, the reduced gap estimate closes
cubically with system size, while the prefactor depends on the dissipation
strength. The corresponding reduced relaxation timescale therefore grows
as $n^3$.

\section{Outlook and Conclusion}

In this study, we have explored whether cooling only a subsystem of a coupled quantum system can lead to effective cooling of the entire system. We derived sufficient conditions under which localized dissipation drives the system to a unique globally attractive zero-excitation state. For multipartite qubit systems with excitation-transfer interactions, we further showed that the zero-forcing property of the interaction graph provides a sufficient condition for this convergence. We illustrated these results numerically for the Heisenberg spin model.

For small system sizes, direct diagonalization of the full Liouvillian agrees to numerical precision with the decay rate obtained from the reduced single-excitation effective Hamiltonian. For $\gamma_1=1$, we analytically find that the corresponding reduced gap estimate scales asymptotically as
\begin{equation*}
\Delta_{\mathrm{red}}(n)
=
\frac{\pi^2}{5n^3}
+
O(n^{-4}),
\end{equation*}
and verify this behavior numerically for systems of up to $n=1000$ qubits. This implies a polynomial $n^3$ scaling of the corresponding reduced relaxation timescale.

Our findings provide a connection between localized dissipative control, excitation transport, and the graph structure of an interacting quantum system. They show that access to only a subset of a system can, under suitable conditions, be sufficient to prepare a global zero-excitation state. When this state coincides with a ground state of the Hamiltonian, the same mechanism provides a route to ground-state cooling. This may be relevant in settings where direct dissipative control of every component of a quantum system is experimentally difficult or resource intensive.

Looking ahead, an important open question is whether the vacuum--single-excitation coherence sector remains the slowest Liouvillian sector for arbitrary system size. Establishing this would clarify when the reduced single-excitation decay rate determines the full Liouvillian gap beyond the small systems tested directly here. It would also be interesting to determine how the asymptotic scaling and its prefactor change with interaction strengths and with the structure of the excitation-transfer graph. More broadly, identifying further classes of Hamiltonians and dissipative processes for which localized cooling leads to global state preparation remains an interesting direction. Additionally, exploring the interplay between cooling methods and error-correction protocols could help clarify the role of localized dissipation in practical quantum architectures.

In conclusion, our results establish sufficient conditions under which localized subsystem cooling leads to global preparation of a unique zero-excitation state, and identify zero forcing as a graph-theoretic criterion that guarantees this behavior for excitation-transfer networks. For the Heisenberg chain considered here, we provide numerical and analytical evidence for this cooling to be efficient.

\bibliography{cooling_bib}
\end{document}